\documentclass[conference]{IEEEtran}
\IEEEoverridecommandlockouts
\usepackage{cite}
\usepackage{amsmath,amsthm,amssymb,amsfonts}
\usepackage{algorithm}
\usepackage{algorithmic}
\usepackage{graphicx}
\usepackage{textcomp}
\usepackage{xcolor}
\usepackage{stfloats}
\usepackage{bm}
\def\BibTeX{{\rm B\kern-.05em{\sc i\kern-.025em b}\kern-.08em
    T\kern-.1667em\lower.7ex\hbox{E}\kern-.125emX}}
\begin{document}

\title{Transmission Scheduling for Multi-loop Wireless Networked Control Based on LQ Cost Offset\\
{\footnotesize}
\thanks{This work was supported by Key R\&D Program of China under Grant 2018YFB1801102, National S\&T Major Project 2017ZX03001011, National Natural Science Foundation of China 61631013, Foundation for Innovative Research Groups of the National Natural Science Foundation of China 61621091, Tsinghua-Qualcomm Joint Project, and Tsinghua University Initiative Scientific Research Program 20193080005.}
}

\author{\IEEEauthorblockN{He Ma\IEEEauthorrefmark{1},
Shidong Zhou\IEEEauthorrefmark{1}, Xiujun Zhang\IEEEauthorrefmark{2} and
Limin Xiao\IEEEauthorrefmark{2}}
\IEEEauthorblockA{\IEEEauthorrefmark{1}Dept. of Electronic Engineering, Tsinghua University, Beijing, China\\
\IEEEauthorrefmark{2}BNRIST, Tsinghua University, Beijing, China
\\
%\IEEEauthorrefmark{3}\IEEEauthorrefmark{4}\textit{BNRIST,}\\
mah17@mails.tsinghua.edu.cn,
\{zhousd,zhangxiujun,xiaolm\}@tsinghua.edu.cn}}

%\author{
%\IEEEauthorblockN{1\textsuperscript{st} He Ma}
%\IEEEauthorblockA{\textit{Dept. of Electronic Engineering,} \\
%\textit{Tsinghua University}\\
%  \\
%mah17@mails.tsinghua.edu.cn}
%\and
%\IEEEauthorblockN{2\textsuperscript{nd} Shidong Zhou}
%\IEEEauthorblockA{\textit{Dept. of Electronic Engineering,} \\
%\textit{Tsinghua University}\\
% Beijing, China \\
%zhousd@tsinghua.edu.cn}
%\and

\maketitle

\begin{abstract}
In this paper, transmission scheduling for multi-loop wireless networked control systems sharing the wireless channel is considered. A linear quadratic cost offset has been proposed to evaluate the performance gap induced by the non-ideal communication. A functional relationship between linear quadratic offset metric and Age of Information has been built up. Based on the offset metric, we come up with an age-based scheduling policy and numerical simulations show that there is a significant improvement compared to the former work.
\end{abstract}

\section{Introduction}
With the emerging new application of networked control, such as inteligent manufacturing, auto-driving, drones, smart grids, telemedicine, etc., there have been increasing demands for wireless networked control systems (WNCSs) due to a complicated industry environment, including but not limited to mobility requirements.
As is well known, although wireless communication has the advantage of supporting mobility, limited radio frequency resource becomes the bottleneck of sufficient, reliable and timely connections between the components of control system, such as sensors, controllers, and actuators.
There are many ways to improve the performance for WNCSs, besides spectrum efficent transmission such as MIMO and high order modulation\cite{garone2011lqg,li2013adaptive}, diversity and advanced channel coding for reliability enhancement\cite{tarable2013anytime,robinson2007sending}, etc., proper dynamic resource allocation or scheduling is also an important way to fulfill the target of WNCSs with multiple control loops which may require wireless access, especially when the radio resource is not sufficient with more and more control loops involved.  

Scheduling strategy should be designed according to the performance metric of the control system. Commonly used control performance metrics include system stability and linear quadratic cost function. Some researchers \cite{kim2003maximum,walsh2002stability} consider the connection between system stability metric and sampling delay in sensors to schedule the sampling interval.  And the stability criterion of control systems with multiple communication modes is given in \cite{hassibi1999control}. These work usually yields stability conditions which offer communication scheme designs with feasible criteria but cannot help to obtain optimal control performance. Therefore, linear quadratic (LQ) cost, which can represent the performance of control systems numerically, is widely used in optimal control and other control scenarios. The difference lies in sampling periods will influence LQ cost in some way, on which a simulation analysis\cite{xu2017lqg} has been concucted. And the sampling status is put into LQ cost criterion and the added cost is optimized\cite{henriksson2015multiple}, however, in essense, the control performance metric and the communication metric are still in an uncoupled form. In \cite{molin2014price,molin2009lqg}, the author takes the particular structure of LQG control into consideration and lists out the specific impact of the scheduling process on LQ cost, therefore optimized the whole LQ cost directly. This method, however, is only valid for LQG control and lacks generality.  
\par

We focus on the detailed relationship between communication metrics and LQ cost which is widely applied and aim to obtain performance gain by scheduling. However, due to the tight coupling of control and communication, it is difficult to directly optimize the LQ cost in general multi-loop systems. Some researchers try to decompose the problem by considering some intermediate metrics that can affect control performance and are easier to evaluate, for example, state estimation error. The effect of transmission delay\cite{huang2019retransmit} on the state estimation error metric has been analysed. 
Due to the sensitivity to state freshness of WNCSs, Age of Information (AoI) \cite{Kaul2012real} which represents the freshness of information is considered to play an important role in WNCSs. 
Different from average delay, AoI process can represent the transmission and scheduling status in communication systems at all times to some extent, and is considered to determine the estimation error in controllers in \cite{champati2019performance,ayan2019age}. Thus, AoI process is an effective tool to help design scheduling policies.
However, state estimation error is more like a communication performance than a direct performance to a control system and therefore is a partial metric of the performance of WNCSs.

Another common way to decompose is that we assume an ideal communication while designing control strategies and then assume unchanged control strategies while designing communication resources scheduling policies\cite{cervin2008scheduling,molin2009lqg}. In this situation, the performance offset between non-ideal communication WNCSs and ideal communication WNCSs can be considered as a LQ control metric as well. In the meantime, it can be extended to the case of control systems with non-optimal cost design, where our offset represents the performance gap induced by the non-ideal communication in the original design. It's beneficial for improving wireless resources scheduling that we study the relationshop between LQ cost offset and communicatin metrics. 

%As WNCSs are systems sensitive  to state freshness, the theory called Age of Information \cite{Kaul2012real} which represents for the freshness of information has brought us a new understanding. AoI process can represent the transmission and scheduling status in communication systems to some extent, and is considered to determine the estimation error metric in controllers in \cite{champati2019performance,ayan2019age}. %. And heuristic scheduling policies based on the AoI process have been proposed. 
%Therefore the problem decomposition is often considered. One common way to decomposite is that we assume ideal communication when designing control strategies and then assume unchanged control strategies  when designing communication resources scheduling policies\cite{ayan2019age,molin2009lqg}. 
%In this situation, the performance offset between non-ideal communication WNCSs and ideal communication WNCSs can be considered as a LQ control metric as well. In the meantime, it can be extended to the case of contorl systems with non-optimal cost design, where our offset represents the performance gap induced by the non-ideal communication in the original design.It's beneficial for improving wireless resources scheduling that we study the relationshop between LQ cost offset and communicatin metrics. 

The main contribution of this paper is as follows: Firstly, we proposed a linear quadratic cost offset metric to evaluate the performance gap induced by the non-ideal communication. Secondly, we made a connection between our proposed control metric and communication metric AoI. What's more, a scheduling policy has been proposed according to our offset metric. A significant performance gain has been obtained compared to existing scheduling policies.

\par 
\section{System Model}
We consider a scenario where $N$ independent linear time-invariant (LTI) subsystems share a common wireless channel, which is depicted in Figure~\ref{fig1}. Each subsystem $i$ consists of a plant, a sensor, an estimator, a controller and an actuator. Actuators and plants are wired to the controller while sensors and estimators communicate through a shared wireless network. At the beginning of the time slot, each subsystem samples its system state.
% and passes the system information, such as AoI and system parameters, to the scheduler. 
The scheduler will allocate resources based on the historical information of all subsystems in every time slot. Then the chosen subsystems will transmit their sampled states to the controller side through the shared wireless channel. We make the following simplifying assumptions for our system model: (i) there are ideal network links between the controller and actuator for each subsystem which are very common in WNCSs, (ii) all the components are assumed to operate and update at the beginning of each time slot synchronously, (iii)The total delay from tx's transmission and execution of system modules is much lower than the time slot, (iv) the wireless channel is an erasure channel with transmission success probability $p$ and (v) in every time slot, each subsystem can employ one communication resource at most.
\begin{figure}[!t]
\centerline{\includegraphics[width=3in]{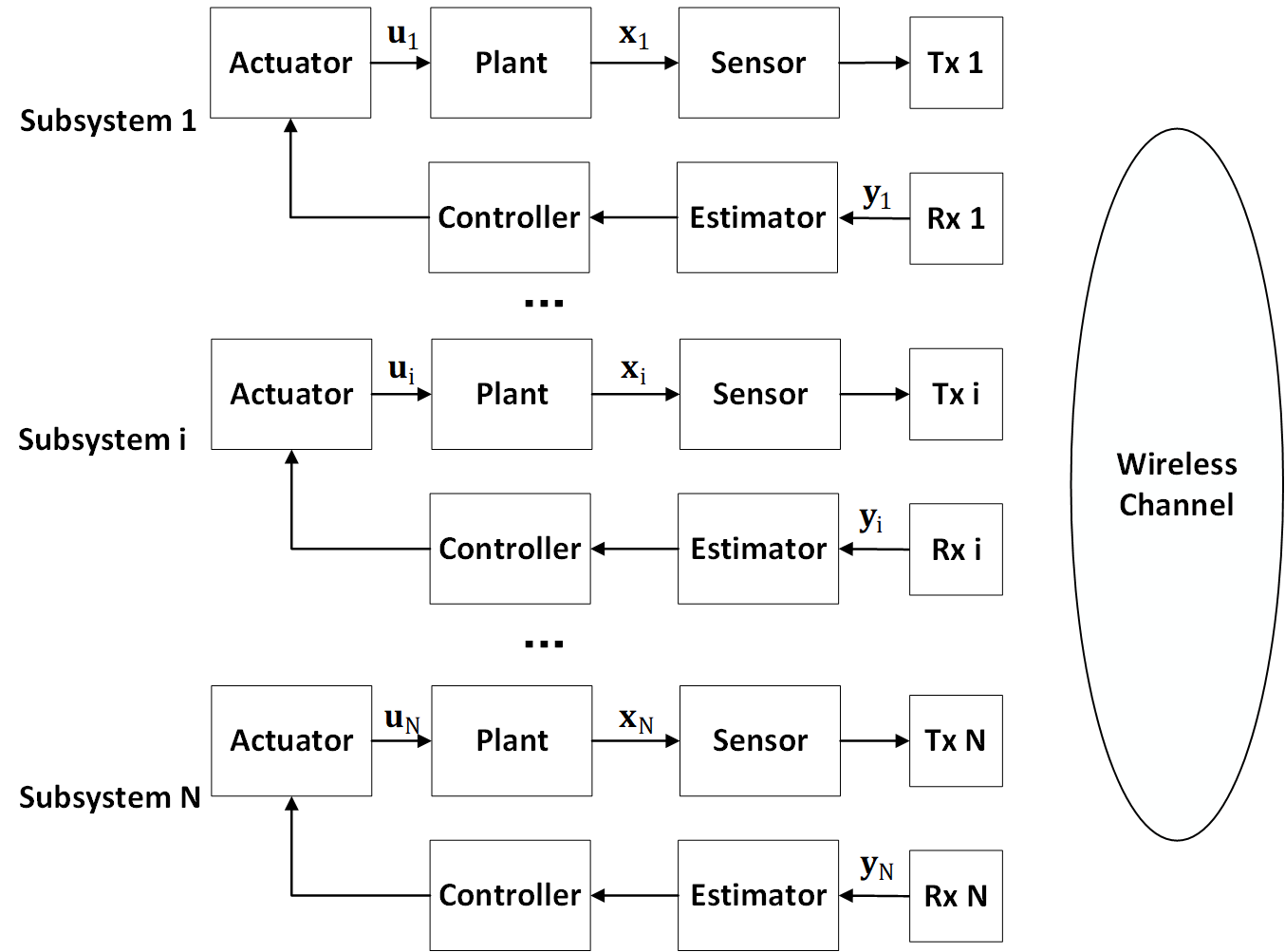}}
\caption{The multi-loop networked control systems model. Actuators and plants are wired to the controller while sensors and estimators communicate through a shared wireless network.}
\label{fig1}
\end{figure}
\subsection{System Setup}
%假设：1.每个系统有同步采样以及拥有相同周期，同步系统。 2.
We consider the dynamic of the $i$-th subsystem at time slot $k$ is represented by the following discrete time LTI model:
\begin{equation} \label{System model}
\mathbf{x}_{i}[k+1]=\mathbf{A}_{i}\mathbf{x}_{i}[k]+\mathbf{B}_{i}\mathbf{u}_{i}[k]+\mathbf{e}_{i}[k]
\end{equation}
with system state $\mathbf{x}_{i}$, control vector $\mathbf{u}_{i}$, system matrix $\mathbf{A}_{i}$ and input matrix $\mathbf{B}_{i}$ determined by the plant. And the system noise $\mathbf{e}_{i}$ is assumed to be an i.i.d. vector which has zero mean Gaussian distribution with diagonal covariance matrix $\mathbf{R}_i$. And each system has a initial state $\mathbf{x}_{i}[0]$ which is known at the controller side.
\subsection{Network Model}
Figure~\ref{fig2} shows how our network model works. In time slot $k$, the scheduler allocates $M$ channel resources to specific subsystems according to the scheduling policy and unscheduled  subsystems remain silent. According to assumption (iv), the channel state distribution of each subsystem is an i.i.d bernoulli distribution with expectation $p$. 
we denote the state observation received by the estimator as $\mathbf{y}_{i}[k]$. $\mathbf{y}_{i}[k]=\mathbf{x}_{i}[k]$ if and only if subsystem $i$ is scheduled to transmit in time slot $k$ and the transmittion is successful, otherwise $\mathbf{y}_{i}[k]=\text{NA}$.
\par
Packets in networked control systems are time sensitive where stale control information may degrade performance of control loops. Thus, the freshness of information, i.e. the time difference bewteen generation and arrival, is one of the most important communication metric. Age of information(AoI) is often used to describe freshness of packets in systems. If at any time slot $k$, the freshest packet received by the controller of system $i$ is generated at time slot $G_i[k]$, the AoI in system $i$ at time slot $k$ is defined as: 
\begin{equation}\label{AOI}
\Delta_i[k] = k - G_i[k]
\end{equation}
%n NCSs, 信息达到时间不是最关键的，信息到达的时间与信息产生时间的差值才是最关键的。因此本文中使用AOI描述，系统信息的新鲜程度，即信息产生时间为时隙k，经过多个时隙的传输，信息成功到达时间为k+delta，则信息的AOI为delta。本文中AOI与调度决策以及随机丢包有直接关系。AOI = \sum ab
From the definition in \eqref{AOI}, we can see that AoI is directly associated with the process of scheduling policy and channel state which will influence information update, and represents the freshness of system state information.
\par
\begin{figure}[!t]
\centerline{\includegraphics[width=3.5in]{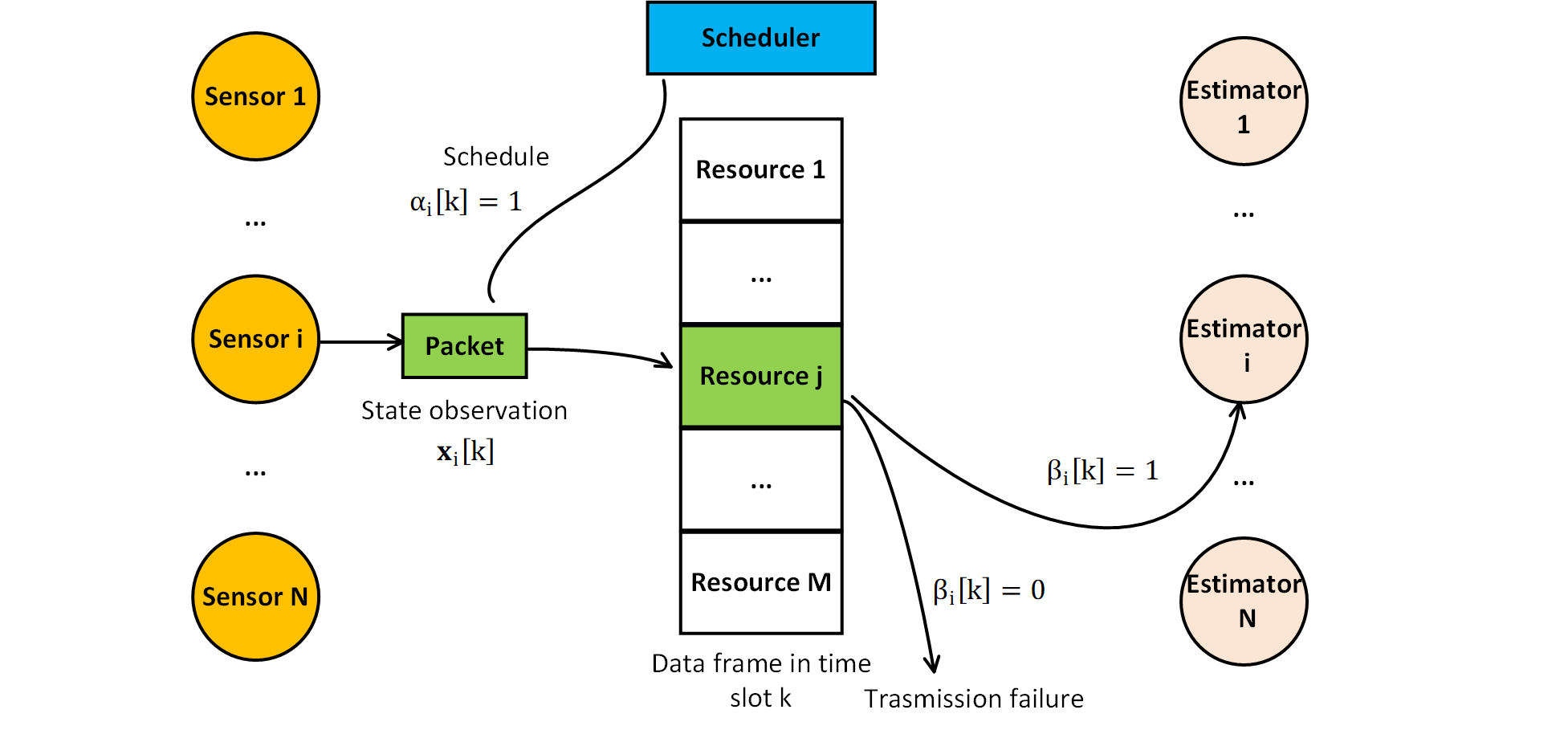}}
\caption{The wireless network model in time slot $k$.}
\label{fig2}
\end{figure}
\subsection{Controller and Estimator Model}
According to assumption (iii), control vector $\mathbf{u}_{i}[k]$ at time slot $k$ is generated by the system state at the same time slot and then fed back to the plants to stabilize corresponding subsystems. In this paper, we suppose all subsystems have taken the deterministic control law as follows:
\begin{equation}\label{Control Law}
\mathbf{u}_{i}[k] = \mathbf{K}_i\hat{\mathbf{x}}_i[k]
\end{equation}
where $\mathbf{K}_i$ is the stationary control gain determined by the control strategy and $\hat{\mathbf{x}}_i[k]$ is the estimation of subsystem $i$'s state at time slot $k$.
\par
The principle of estimation is as follows: If the packet of subsystem $i$ in time slot $k$ has been received successfully, the estimator can directly use the state observation $\mathbf{y}_i[k]$ as the estimation, otherwise it will use the estimation of the last time slot $\hat{\mathbf{x}}_i[k-1]$ to iterate.
$\alpha_i[k]$ and $\beta_i[k]$ are used to describe the scheduling result and channel state of subsystem $i$ on time slot $k$ respectively. Namely, $\alpha_i[k]=1$ means subsystem $i$ can transmit its state observation at time slot $k$ based on the scheduling result and $\beta_i[k]=1$ means subsystem $i$ can successfully transmit its data at time slot $k$. If $\alpha_i[k]=1$ and $\beta_i[k]=1$, we have $\mathbf{y}_{i}[k] = \alpha_i[k]\beta_i[k]\mathbf{x}_{i}[k]$.
%We consider the information set known by the estimator of subsystem $i$ at time slot $k$ as $\mathbb{I}_{e,i}[k]=\{\mathbf{y}_i[k],\hat{\mathbf{x}}_i[k-1],\alpha_i(k),\beta_i(k)\}$, and state estimation of the subsystem $i$ as $\hat{\mathbf{x}}_i[k] = \mathbb{E}[\mathbf{x}_i[k]|\mathbb{I}_{e,i}[k]]$. 
Thus, the dynamics of the estimator can be written as:
%这个地方需要在思考一下怎么写
\begin{equation}\label{estimator}
\begin{aligned}
\hat{\mathbf{x}}_i[k] & = \alpha_i(k)\beta_i(k)\mathbf{x}_i[k]+\\
&(1-\alpha_i(k)\beta_i(k))(\mathbf{A}_{i}\hat{\mathbf{x}}_i[k-1]+\mathbf{B}_{i}\mathbf{K}_i\hat{\mathbf{x}}_i[k-1])
\end{aligned}
\end{equation}
\par
Then by combining \eqref{System model}, \eqref{Control Law} and \eqref{estimator}, the whole system dynamics is:
\begin{equation}\label{Complete Dynamics}
\left\{ 
\begin{array}{lr} 
\mathbf{x}_{i}[k+1]=\mathbf{A}_{i}\mathbf{x}_{i}[k]+\mathbf{B}_{i}\mathbf{u}_{i}[k]+\mathbf{e}_i[k], & \\ 
\mathbf{u}_{i}[k] = \mathbf{K}_i\hat{\mathbf{x}}_i[k], &\\ 
\hat{\mathbf{x}}_i[k] = \alpha_i(k)\beta_i(k)\mathbf{x}_i[k]+ &\\
(1-\alpha_i(k)\beta_i(k))(\mathbf{A}_{i}\hat{\mathbf{x}}_i[k-1]+\mathbf{B}_{i}\mathbf{K}_i\hat{\mathbf{x}}_i[k-1]) &
\end{array} 
\right. 
\end{equation}
\section{Control Performance Metrics}
We first consider Linear Quadratic average-cost criterion \cite{molin2014price} which can describe the control performance of each subsystem as:
\begin{equation}\label{Linear Quadratic average-cost criterion 1}
\begin{aligned}
& J_i & =  \lim_{T \to +\infty}\frac{1}{T}\sum_{k=1}^{T}\mathbb{E}\{{\mathbf{x}_{i}}^\mathbf{T}[k]\mathbf{Q}_{i}\mathbf{x}_{i}[k]+{\mathbf{u}_{i}}^\mathbf{T}[k]\mathbf{P}_{i}{\mathbf{u}_{i}}[k]\}
\end{aligned}
\end{equation}
where $\mathbf{Q}_{i}$ represents the state weights matrix of subsystem $i$ and $\mathbf{P}_{i}$ represents the control energy weights matrix of subsystem $i$.

And the sum cost is considered as the overall criterion to evaluate the performance of all subsystems
\begin{equation}\label{Linear Quadratic average-cost criterion 2}
\begin{aligned}
J & =  \sum_{i=1}^{N}J_i \\
& =\lim_{T \to +\infty}\frac{1}{T}\sum_{i=1}^{N}\sum_{k=1}^{T}\mathbb{E}\{{\mathbf{x}_{i}}^\mathbf{T}[k]\mathbf{Q}_{i}\mathbf{x}_{i}[k]+{\mathbf{u}_{i}}^\mathbf{T}[k]\mathbf{P}_{i}{\mathbf{u}_{i}}[k]\}
\end{aligned}
\end{equation}
%metric defination show Introduction of Age of Information and then the derivation of our metric as a function of AOI
%我们使用LQ cost来评价该控制系统的性能，当控制策略给定的时候，会有一个对应的LQ cost，非理想通信传输会影响控制效果，使系统与既定设计产生偏差。我们将这种偏差定义为一种度量，衡量实际系统与目标系统的性能偏差。

Once the control strategy is determined, there will be a corresponding LQ cost. Non-ideal communication will affect the control effect and make systems deviate from the established design. We define this performance deviation, which we called as Linear Quadratic cost offset, as a control metric of the performance offset between ideal and non-ideal communication WNCSs. If we set $\mathbf{x}_{\text{opt},i}[k]$ and $\mathbf{u}_{\text{opt},i}[k]$ as the ideal system state and control vector when there are zero packet-drop and unlimited slot resources in time interval $[k-\Delta_i[k],k]$, the error between non-ideal state and ideal state is $\overline{\mathbf{x}}_i[k] = \mathbf{x}_{i}[k] - \mathbf{x}_{\text{opt},i}[k] $. Noted that when the system observation has been transmitted successfully in time slot $k-\Delta_i[k]$, we have $\overline{\mathbf{x}}_i[k-\Delta_i[k]] = 0$. Therefore the Linear Quadratic cost offset is defined as:
\begin{equation}\label{control cost gap defination}
\begin{aligned}
J_{\text{offset}} = &\lim_{T \to +\infty}\frac{1}{T}\sum_{i=1}^{N}\sum_{k=1}^{T}\mathbb{E}\{{\mathbf{x}_i}^\mathbf{T}[k]\mathbf{Q}_{i}\mathbf{x}_i[k]+{\mathbf{u}_{i}}^\mathbf{T}[k]\mathbf{P}_{i}{\mathbf{u}_{i}}[k]\\
&-{\mathbf{x}_{\text{opt},i}}^\mathbf{T}[k]\mathbf{Q}_{i}\mathbf{x}_{\text{opt},i}[k]+{\mathbf{u}_{\text{opt},i}}^\mathbf{T}[k]\mathbf{P}_{i}{\mathbf{u}_{\text{opt},i}}[k]\}
%= &\lim_{T \to +\infty}\frac{1}{T}\sum_{i=1}^{N}\sum_{k=1}^{T}\mathbb{E}\{{\Delta \mathbf{x}_i}^\mathbf{T}[k]\mathbf{Q}_{i}\Delta \mathbf{x}_i[k]+\Delta{\mathbf{u}_{i}}^\mathbf{T}[k]\mathbf{P}_{i}\Delta{\mathbf{u}_{i}}[k]\}
\end{aligned}
\end{equation}
This metric actually represents the average accumulated performance gap induced by non-ideal communication as time proceeded. 
%我们考虑衡量无通信影响的理想状态和有通信影响的状态之间的差距，我们用lq度量差作为衡量的度量，考虑k时刻的第i个系统的aoi为，状态为，假设[k-delta,k]时刻内通信理想，无丢包且资源无限，则理想状态记为，则LQ度量差定义为。
\par
\section{Analysis of Control Cost Offset and AoI}
In this section, we will analyse the relationship between cost offset metric and Age of Information. 
%\newcounter{TempEqCnt} % 创建临时变量TempEqCnt
%\setcounter{TempEqCnt}{\value{equation}} % 将当前公式序号 赋给TempEqCnt
%\setcounter{equation}{12} % 当前公式序号变为x，x等于长公式应有的序号减1.
%\setcounter{equation}{\value{TempEqCnt}} % 把TempEqCnt中存的公式序号赋回给当前公式序号
\newtheorem{lemma}{Lemma}
\begin{lemma} \label{lemma0}
For the WNCSs described in \eqref{estimator} and \eqref{Complete Dynamics}, the state $\mathbf{x}_i[k]$ in subsystem $i$ at time slot $k$ with nonzero AoI $\Delta_i[k]$ has the following form:
\begin{equation}\label{estimator2}
\begin{aligned}
\mathbf{x}_i[k] & = (\mathbf{A}_i+\mathbf{B}_i\mathbf{K}_i)^{\Delta_i[k]}\mathbf{x}_{i}[k-\Delta_i[k]]\\
&+\sum_{j=0}^{\Delta_i[k]-1}{\mathbf{A}_i}^{j}\mathbf{e}_{i}[k-j-1]\\
\end{aligned}
\end{equation}
\end{lemma}
\begin{proof}
Nonzero $\Delta_i[k]$ means that estimator has successfully received the actual state $\mathbf{x}_{i}[k-\Delta_i[k]]$ at time slot $k-\Delta_i[k]$ and from that time slot on this subsystem did not get any state updated. Thus from~\eqref{estimator}, we can easily obtain the state estimation $\hat{\mathbf{x}}_i[k]$ as:
\begin{equation}\nonumber
\begin{aligned}
&\hat{\mathbf{x}}_i[k]=(\mathbf{A}_i+\mathbf{B}_i\mathbf{K}_i)^{\Delta_i[k]}\mathbf{x}_{i}[k-\Delta_i[k]]\\
\end{aligned}
\end{equation}
Applying the system dynamics~\eqref{Complete Dynamics} and the estimation above, we can get the state $\mathbf{x}_{i}[k]$ iteratively.
\begin{equation}\nonumber
\begin{aligned}
\mathbf{x}_{i}[k-\Delta_i[k]+1]&=(\mathbf{A}_i+\mathbf{B}_i\mathbf{K}_i)\mathbf{x}_{i}[k-\Delta_i[k]]\\
&+\mathbf{e}_i[k-\Delta_i[k]]\\ 
\mathbf{x}_{i}[k-\Delta_i[k]+2]&=\mathbf{A}_{i}\mathbf{x}_{i}[k-\Delta_i[k]+1]\\
+\mathbf{B}_i\mathbf{K}_i(\mathbf{A}_i&+\mathbf{B}_i\mathbf{K}_i)\mathbf{x}_{i}[k-\Delta_i[k]]+\mathbf{e}_i[k-\Delta_i[k]+1]\\ 
&=(\mathbf{A}_i+\mathbf{B}_i\mathbf{K}_i)^{2}\mathbf{x}_{i}[k-\Delta_i[k]]\\
&+\mathbf{A}_{i}\mathbf{e}_i[k-\Delta_i[k]]+\mathbf{e}_i[k-\Delta_i[k]+1]\\
&......\\
\mathbf{x}_{i}[k-\Delta_i[k]+l]&=(\mathbf{A}_i+\mathbf{B}_i\mathbf{K}_i)^{l}\mathbf{x}_{i}[k-\Delta_i[k]]\\
&+\sum_{j=0}^{l-1}{\mathbf{A}_i}^{j}\mathbf{e}_{i}[k-j-1]\\
&......\\
\mathbf{x}_{i}[k]&=(\mathbf{A}_i+\mathbf{B}_i\mathbf{K}_i)^{\Delta_i[k]}\mathbf{x}_{i}[k-\Delta_i[k]]\\
&+\sum_{j=0}^{\Delta_i[k]-1}{\mathbf{A}_i}^{j}\mathbf{e}_{i}[k-j-1]\\
\end{aligned}
\end{equation}
\end{proof}
The offset between non-ideal state and ideal state can be derived from system dynamics and Lemma~\ref{lemma0}, leading to the following lemma.
\begin{lemma} \label{lemma1}
For NCSs described in \eqref{Complete Dynamics}, the state offset in system $i$ at time slot $k$ can be written as:
\begin{equation}\label{state error}
\begin{aligned}
\overline{\mathbf{x}}_i[k] = \sum_{j=0}^{\Delta_i[k]-1}({\mathbf{A}_i}^{j}-{(\mathbf{A}_i+\mathbf{B}_i\mathbf{K}_i)}^{j})\mathbf{e}_i[k-j-1]
\end{aligned}
\end{equation}
\end{lemma}
\begin{proof}
Consider the ideal conditon where $\alpha$ and $\beta$ are always one and the non-ideal condition where systems have nonzero  AOI, $\Delta_i[k]$, caused by scheduling and packet loss. The offset accumulates from time slot $k-\Delta_i[k]$ to time slot $k$. Using Lemma~\ref{lemma0} and \eqref{Complete Dynamics}, the state offset is:
\begin{equation}\nonumber
\begin{aligned}
\overline{\mathbf{x}}_i[k] &= \mathbf{x}_{i}[k] - \mathbf{x}_{\text{opt},i}[k] \\
& = (\mathbf{A}_i+\mathbf{B}_i\mathbf{K}_i)^{\Delta_i[k]}\mathbf{x}_{i}[k-\Delta_i[k]]\\
&+\sum_{j=0}^{\Delta_i[k]-1}{\mathbf{A}_i}^{j}\mathbf{e}_{i}[k-j-1]\\
&-(\mathbf{A}_i+\mathbf{B}_i\mathbf{K}_i)^{\Delta_i[k]}\mathbf{x}_{i}[k-\Delta_i[k]]\\
& -\sum_{j=0}^{\Delta_i[k]-1}{(\mathbf{A}_i+\mathbf{B}_i\mathbf{K}_i)}^{j}\mathbf{e}_{i}[k-j-1]\\
& =  \sum_{j=0}^{\Delta_i[k]-1}({\mathbf{A}_i}^{j}-{(\mathbf{A}_i+\mathbf{B}_i\mathbf{K}_i)}^{j})\mathbf{e}_i[k-j-1]
\end{aligned}
\end{equation}
\end{proof}
It can be easily seen from Lemma \ref{lemma1} that, when $\Delta_i[k] = 1$, the state offset is zero as well. In this situation, performance of this system doesn't get degraded beacuse the state in time slot $k$ is generated from the state in time slot $k-1$, which is successfully updated. Transmission status in time slot $k$ will affect system states in time slot $k+1$. So scheduling in current time slot will influence the performance of our WNCSs in the next time slot.
\par
\newtheorem{remark}{Remark}
\begin{remark} \label{remark1}
As ${e}_i[k]$ is i.i.d. zero-mean gaussian noise, it is obvious that the state error $\overline{\mathbf{x}}_i[k]$ which is the linear combination of Gaussian noise, has a Gaussian distribution and zero expectation, $\mathbb{E}[\overline{\mathbf{x}}_i[k]] = 0$.
\end{remark}

\par
By using previous lemmas and the definition of LQ cost offset metric, we can get the following result about our metric and AoI.
\begin{lemma} \label{lemma2}
For WNCSs described in \eqref{Complete Dynamics} and cost offset metric described in \eqref{control cost gap defination}, the time average Linear Quadratic cost offset in subsystem $i$ at time slot $k$ can be written as:
\begin{equation}\label{LQ error}
\begin{aligned}
J_{\text{offset}} & = \lim_{T \to +\infty}\frac{1}{T}\sum_{i=1}^{N}\sum_{k=1}^{T}\sum_{j=0}^{\Delta_i[k]-1}\mathbf{Tr}\{\left[\mathbf{Q}_{i}+{\mathbf{K}_{i}}^\mathbf{T}\mathbf{P}_{i}\mathbf{K}_{i}\right]\\
&{({\mathbf{A}_i}^{j}-{(\mathbf{A}_i+\mathbf{B}_i\mathbf{K}_i)}^{j})}\mathbf{R}_i{({\mathbf{A}_i}^{j}-{(\mathbf{A}_i+\mathbf{B}_i\mathbf{K}_i)}^{j})}^T\}\\
\end{aligned}
\end{equation}
\end{lemma}
\begin{proof}
Using lemma~\ref{lemma1} and elementary matrix caculation, we can derive the following form of the LQ offset.
\begin{equation}\nonumber
\begin{aligned}
J_{\text{offset}} 
& = \lim_{T \to +\infty}\frac{1}{T}\sum_{i=1}^{N}\sum_{k=1}^{T}\mathbb{E}\{{\mathbf{x}_i}^\mathbf{T}[k](\mathbf{Q}_{i}+{\mathbf{K}_{i}}^\mathbf{T}\mathbf{P}_{i}\mathbf{K}_{i})\mathbf{x}_i[k]-\\
&{\mathbf{x}_{\text{opt},i}}^\mathbf{T}[k](\mathbf{Q}_{i}+{\mathbf{K}_{i}}^\mathbf{T}\mathbf{P}_{i}\mathbf{K}_{i})\mathbf{x}_{\text{opt},i}[k]\}\\
& = \lim_{T \to +\infty}\frac{1}{T}\sum_{i=1}^{N}\sum_{k=1}^{T}\mathbb{E}\{{\mathbf{x}_i}^\mathbf{T}[k](\mathbf{Q}_{i}+{\mathbf{K}_{i}}^\mathbf{T}\mathbf{P}_{i}\mathbf{K}_{i})\mathbf{x}_i[k]-\\
&{(\mathbf{x}_{i}+\overline{\mathbf{x}}_i[k])}^\mathbf{T}(\mathbf{Q}_{i}+{\mathbf{K}_{i}}^\mathbf{T}\mathbf{P}_{i}\mathbf{K}_{i})(\mathbf{x}_{i}+\overline{\mathbf{x}}_i[k])\}\\
& = \lim_{T \to +\infty}\frac{1}{T}\sum_{i=1}^{N}\sum_{k=1}^{T}\mathbb{E}\{{\overline{\mathbf{x}}_i}^\mathbf{T}[k](\mathbf{Q}_{i}+{\mathbf{K}_{i}}^\mathbf{T}\mathbf{P}_{i}\mathbf{K}_{i})\overline{\mathbf{x}}_i[k]\}\\
%& = \lim_{T \to +\infty}\frac{1}{T}\sum_{i=1}^{N}\sum_{k=1}^{T}\mathbb{E}\{\mathbf{Tr}[\Delta\mathbf{x}_i[k]{\Delta\mathbf{x}_i}^\mathbf{T}[k](\mathbf{Q}_{i}+{\mathbf{K}_{i}}^\mathbf{T}\mathbf{P}_{i}\mathbf{K}_{i})]\}\\
& = \lim_{T \to +\infty}\frac{1}{T}\sum_{i=1}^{N}\sum_{k=1}^{T}\mathbb{E}\{\mathbf{Tr}\{[{ \sum_{j=0}^{\Delta_i[k]-1}(\mathbf{Q}_{i}+{\mathbf{K}_{i}}^\mathbf{T}\mathbf{P}_{i}\mathbf{K}_{i})}\\
&({\mathbf{A}_i}^{j}-{(\mathbf{A}_i+\mathbf{B}_i\mathbf{K}_i)}^{j}){\mathbf{e}_i[k-j-1]{\mathbf{e}_i[k-j-1]}^\mathbf{T}}\\
&{{({\mathbf{A}_i}^{j}-{(\mathbf{A}_i+\mathbf{B}_i\mathbf{K}_i)}^{j})}^\mathbf{T}}\}\}\\ 
& = \lim_{T \to +\infty}\frac{1}{T}\sum_{i=1}^{N}\sum_{k=1}^{T}\sum_{j=0}^{\Delta_i[k]-1}\mathbf{Tr}\{\left[\mathbf{Q}_{i}+{\mathbf{K}_{i}}^\mathbf{T}\mathbf{P}_{i}\mathbf{K}_{i}\right]\\
&{({\mathbf{A}_i}^{j}-{(\mathbf{A}_i+\mathbf{B}_i\mathbf{K}_i)}^{j})}\mathbf{R}_i{({\mathbf{A}_i}^{j}-{(\mathbf{A}_i+\mathbf{B}_i\mathbf{K}_i)}^{j})}^T\}\\
\end{aligned}
\end{equation}
where $\mathbf{R}_i$ is the covariance matrix of Gaussian noise in subsystem $i$.
\end{proof}
As we can see from lemma \ref{lemma2}, the linear quadratic cost offset is a kind of non-linear function of AoI process, which means that the process of Age of Information and the system parameters determine the performance offset of our WNCSs.
\par
%\newtheorem{corollary}{Corollary}
%\begin{corollary} \label{Corollary1}
%When every subsystem $[\mathbf{A}_i,\mathbf{B}_i\mathbf{K}_i]$ and $[\mathbf{A}_i]$ is stabilized, i.e. $||\mathbf{A}_i+\mathbf{B}_i\mathbf{K}_i||_2<1$ and $||\mathbf{A}_i||_2<1$, and AoIs of all subsystems approach to $+\infty$, we %get \eqref{property} as an upper bound of the control cost offset in this specific scenario.
%\end{corollary}
%\begin{proof}
%When the AoIs of all systems approach to infinity, we can untilize summation of geometric matrices for simplification and then get an upper bound of proposed control cost offset when systems are all stabilizable. 
%\end{proof}
%\begin{figure*}[!t]
%\begin{equation}\label{property}
%\begin{aligned}
%J_{\text{offset}} &\leq \sum_{i=1}^{N}\mathbf{Tr}\{\left[\mathbf{Q}_{i}+{\mathbf{K}_{i}}^\mathbf{T}\mathbf{P}_{i}\mathbf{K}_{i}\right]{\left[{(\mathbf{I}-\mathbf{A}^T_i \mathbf{A}_i)}^{-1}+{(\mathbf{I}-(\mathbf{A}_i+\mathbf{B}%_i\mathbf{K}_i)^T (\mathbf{A}_i+\mathbf{B}_i\mathbf{K}_i))}^{-1}-2{(\mathbf{I}-\mathbf{A}^T_i (\mathbf{A}_i+\mathbf{B}_i\mathbf{K}_i))}^{-1}\right]}\mathbf{R}_i\}\\
%\end{aligned}
%\end{equation}
%\end{figure*}
\section{An Age-based Scheduling Method}
%\begin{equation}\nonumber
%\begin{aligned}
%1
%\end{aligned}
%\end{equation}
\par
\par
%非理想通信信道会带来资源上的限制，合理的通信调度可以提高资源利用率，优化资源配置。根据上面的理论，
Non-ideal communication network means constraints on communication resources so that transmission needs for all subsystems cannot be satified simultaneously. Employing effective scheduling policy can optimize resources allocation and therefore improve WNCSs performance. Suppose the information set availabe at the scheduler is $\mathbb{I}_{S}[k]=\{{\{{\mathbf{A}}_i\}}_{i=1}^{N},{\{{\mathbf{B}}_i\}}_{i=1}^{N},{\{{\mathbf{K}}_i\}}_{i=1}^{N},{\{{\mathbf{R}}_i\}}_{i=1}^{N},{\{\Delta_i[k]\}}_{i=1}^{N}\}$. Based on the theory above, we propose an AoI-based scheduling problem which aims to minimize our proposed control cost offset:
\setcounter{equation}{12}
\begin{equation}\label{op1}
\begin{aligned}
\min_{{\alpha[k]}} \quad & J_{\text{offset}}\\
\mbox{s.t.}\quad
&\sum_{i=1}^{N}{\alpha_i}[k] \leq M\\
\end{aligned}
\end{equation}
It's obvious that this is an infinite horizon integer programming. It's hard to solve this infinite horizon problem in polynomial time, so we come up with a greedy solution to solve it. As scheduling in current time slot will influence the performance of our WNCSs in the next time slot. we try to minimize total control cost offset contribution of all sub-systems in every time slot $k$ by picking M sub-systems which have maximum expected control cost offsets $J_{\text{offset},i}[k]$:
\begin{equation}\label{control cost gap}
\begin{aligned}
J_{\text{offset},i}[k] 
& = \sum_{j=0}^{\Delta_i[k]}\mathbb{E}\{\mathbf{Tr}\{\left[\mathbf{Q}_{i}+{\mathbf{K}_{i}}^\mathbf{T}\mathbf{P}_{i}\mathbf{K}_{i}\right]\\
&{({\mathbf{A}_i}^{j}-{(\mathbf{A}_i+\mathbf{B}_i\mathbf{K}_i)}^{j})}\mathbf{R}_i{({\mathbf{A}_i}^{j}-{(\mathbf{A}_i+\mathbf{B}_i\mathbf{K}_i)}^{j})}^T\}\}\\
\end{aligned}
\end{equation}
And the greedy optimization problem is:
\begin{equation}\label{op2}
\begin{aligned}
\max_{{\alpha[k]}} \quad &  \sum_{i=1}^{N}\alpha_i[k]J_{\text{offset},i}[k]  \\
\mbox{s.t.}\quad
&\sum_{i=1}^{N}{\alpha_i}[k] \leq M
\end{aligned}
\end{equation}
The scheduling policy is given in algorithm \ref{Framework}.
\begin{algorithm}[!t] 
\caption{ Framework of our scheduling policy for WNCSs.} 
\label{Framework} 
\begin{algorithmic}[1] %这个1 表示每一行都显示数字
\REQUIRE ~~\\ %算法的输入参数：Input
The parameters of all N subsystems, ${\{{\mathbf{A}}_i\}}_{i=1}^{N}$, ${\{{\mathbf{B}}_i\}}_{i=1}^{N}$, ${\{{\mathbf{K}}_i\}}_{i=1}^{N}$;
The set of Age of information of all N subsystems at time slot $k$, ${\{\Delta_i[k]\}}_{i=1}^{N}$;\\
The covariance matrices of noise of all N subsystems, ${\{{\mathbf{R}}_i\}}_{i=1}^{N}$;
\ENSURE ~~\\ %算法的输出：Output
The set of scheduling states,  $\bm{\alpha}[k]$;
\STATE Initializing scheduling states with $\alpha_i[k] = 0$; 
\label{code:fram:Initialize}
\STATE Calculating expected control cost offsets $J_{\text{offset},i}[k]$ with the help of $\mathbf{A}_i$, $\mathbf{B}_i$, $\mathbf{K}_i$, $\mathbf{R}_i$ and $\Delta_i[k]$; 
\label{ code:fram:Calculate }%对此行的标记，方便在文中引用算法的某个步骤
\STATE Extracting M sub-systems with the largest offsets $J_{\text{offset},i}[k]$ and put their scheduling states with value 1; 
\label{code:fram:Extract}
\RETURN $\bm{\alpha}[k]$; %算法的返回值
\end{algorithmic}
\end{algorithm}

In the next section, we show that our proposed scheduling policy has performance advantages in LQ cost offset metric compared to other existed scheduling policies.

\section{Numerical Results}
\par
To illustrate our proposed scheduling policy, we now continue with giving simulations. 
\subsection{Simulation Setup}
Empiric cost is used to approximately represent our proposed cost offset criterion and average-cost criterion for ease of calculation, and can be calculated by:
\begin{equation}\label{empiric cost}
\begin{aligned}
&J_{\text{offset}} \approx \frac{1}{T}\sum_{i=1}^{N}\sum_{k=1}^{T}\sum_{j=0}^{\Delta_i[k]-1}\mathbf{Tr}\{\left[\mathbf{Q}_{i}+{\mathbf{K}_{i}}^\mathbf{T}\mathbf{P}_{i}\mathbf{K}_{i}\right]\\
&{({\mathbf{A}_i}^{j}-{(\mathbf{A}_i+\mathbf{B}_i\mathbf{K}_i)}^{j})}\mathbf{R}_i{({\mathbf{A}_i}^{j}-{(\mathbf{A}_i+\mathbf{B}_i\mathbf{K}_i)}^{j})}^T\}\\
&J \approx  \frac{1}{T}\sum_{i=1}^{N}\sum_{k=1}^{T}\mathbb{E}\{{\mathbf{x}_{i}}^\mathbf{T}[k]\mathbf{Q}_{i}\mathbf{x}_{i}[k]+{\mathbf{u}_{i}}^\mathbf{T}[k]\mathbf{P}_{i}{\mathbf{u}_{i}}[k]\}
\end{aligned}
\end{equation}
where $T = 5000$ is the simulation time length in our work.
\par
We perform simulations where $N=8$ subsystems sharing the wireless network and scheduler makes full use of historical information $\mathbb{I}_{S}[k]$ to decide which sub-systems should transmit their states. The parameters and control laws of all eight subsystems we considered are shown in Table I and Table II. Furthermore, the initial state is given with the deterministic initial value $\mathbf{x}_i[0] = {[1,1]}^T$.
\par
% A_i加耦合
% 轮询 用差分累积
%  如果出现丢包 累积项没有贡献 直线斜率 横轴是离散的 如果服务不成功 进行一定的补偿，能不能给出一个目标值 调度次数线性增长和没达到线性增长缺多少的  
% AOI这套理论每次传都传完整状态空间，能不能部分状态反馈 ，每次只传相同的维度，由于耦合
%部分状态预测中AOI的作用
%多系统 n时隙n个系统，是一个时隙传完一个系统，还是n个时隙传完n个系统。
\par
\par
To evaluate our policy, we introduce four different scheduling policies to compare.
\begin{itemize}
\item \textbf{AOI-minimal policy} In every time slot, the scheduler chooses $M$ subsystems with the maximal AOI to transmit. As the conclusion in \cite{kadota2018scheduling}, this greedy scheduling is age-optimal.
\item \textbf{Estimation error minimal policy} In every time slot, the scheduler chooses $M$ subsystems with the maximal quadratic state estimation error\cite{ayan2019age} to transmit.
% Estimation error- minimal !!!
\item \textbf{Round-Robin policy} In every time slot, the scheduler allocates the communication resources sequentially.
\item \textbf{Random policy} In every time slot, the scheduler randomly chooses $M$ subsystems to transmit.
\end{itemize}
\par
\subsection{Time sensitive systems}
Firstly, we consider a time sensitive system where most subsystems have a system state matrix with eigenvalues of which real part is larger than one, which means that they will suffer perfermance loss if there is no fresh state information updated. In the mean time, different systems will bring different contributions to the control cost in the form of different $\mathbf{Q}_i$.  The parameters and control laws of such system are showed in Table I.
\begin{table}[!t]
\caption{System Parameters and control laws of time sensitive systems}
\begin{center}
\begin{tabular}{|c|c|c|c|c|c|c|}
\hline
$i$&$\mathbf{A}_i$&$\mathbf{B}_i$&$\mathbf{Q}_i$&$\mathbf{P}_i$&$\mathbf{R}_i$&$\mathbf{K}_i$\\
\hline
1& $1.1*\mathbf{S}$&$\mathbf{I}$&$100*\mathbf{I}$&$\mathbf{I}$&$0.25*\mathbf{I}$&$-0.2*\mathbf{I}$  \\
\hline
2& $1.1*\mathbf{S}$&$\mathbf{I}$&$100*\mathbf{I}$&$\mathbf{I}$&$0.25*\mathbf{I}$&$-0.3*\mathbf{I}$  \\
\hline
3& $1.2*\mathbf{S}$&$\mathbf{I}$&$10*\mathbf{I}$&$\mathbf{I}$&$0.25*\mathbf{I}$&$-0.4*\mathbf{I}$  \\
\hline
4& $1.2*\mathbf{S}$&$\mathbf{I}$&$8*\mathbf{I}$&$\mathbf{I}$&$0.25*\mathbf{I}$&$-0.6*\mathbf{I}$  \\
\hline
5& $1.3*\mathbf{S}$&$\mathbf{I}$&$6*\mathbf{I}$&$\mathbf{I}$&$0.25*\mathbf{I}$&$-0.8*\mathbf{I}$  \\
\hline
6& $1.3*\mathbf{S}$&$\mathbf{I}$&$4*\mathbf{I}$&$\mathbf{I}$&$0.25*\mathbf{I}$&$-1*\mathbf{I}$  \\
\hline
7& $1.4*\mathbf{S}$&$\mathbf{I}$&$2*\mathbf{I}$&$\mathbf{I}$&$0.25*\mathbf{I}$&$-1.2*\mathbf{I}$  \\
\hline
8& $1.4*\mathbf{S}$&$\mathbf{I}$&$\mathbf{I}$&$\mathbf{I}$&$0.25*\mathbf{I}$&$-1,2*\mathbf{I}$  \\
\hline 
\multicolumn{6}{l}{
where $\mathbf{I}=\begin{aligned}
\left[
\begin{array}{cc}
1&0\\
0&1
\end{array}
\right]
\end{aligned}$  and $\mathbf{S}=\begin{aligned}
\left[
\begin{array}{cc}
1&0.2\\
-0.2&1
\end{array}
\right]
\end{aligned}$}
\end{tabular}
\label{tab1}
\end{center}
\end{table}

\begin{figure}[!t]
\centerline{\includegraphics[width=2.5in]{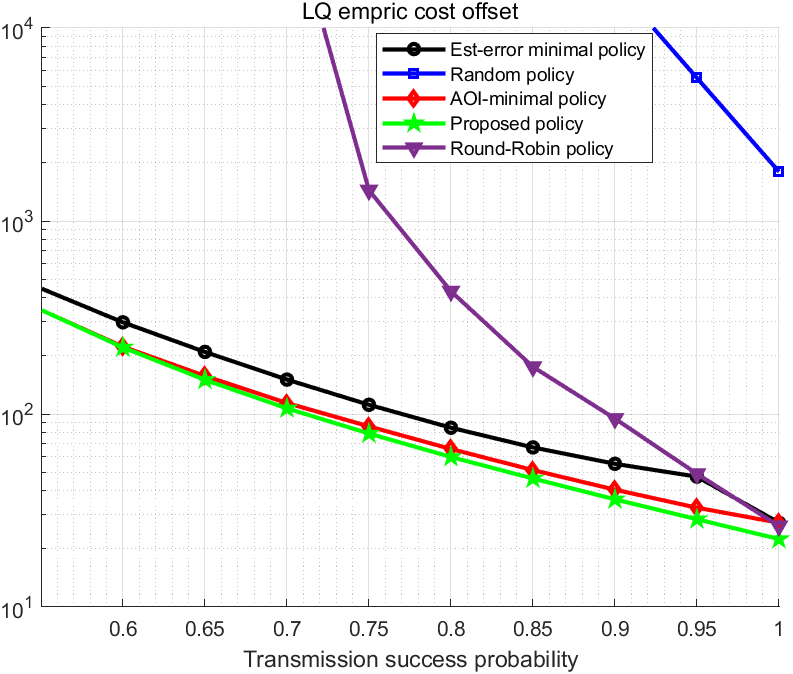}}
\caption{Cost offset comparison of the five types of scheduling policies with different channel states $p$ and $M=3$ wireless resources in time sensitive systems.}
\label{fig3}
\end{figure}

\begin{figure}[!t]
\centerline{\includegraphics[width=2.5in]{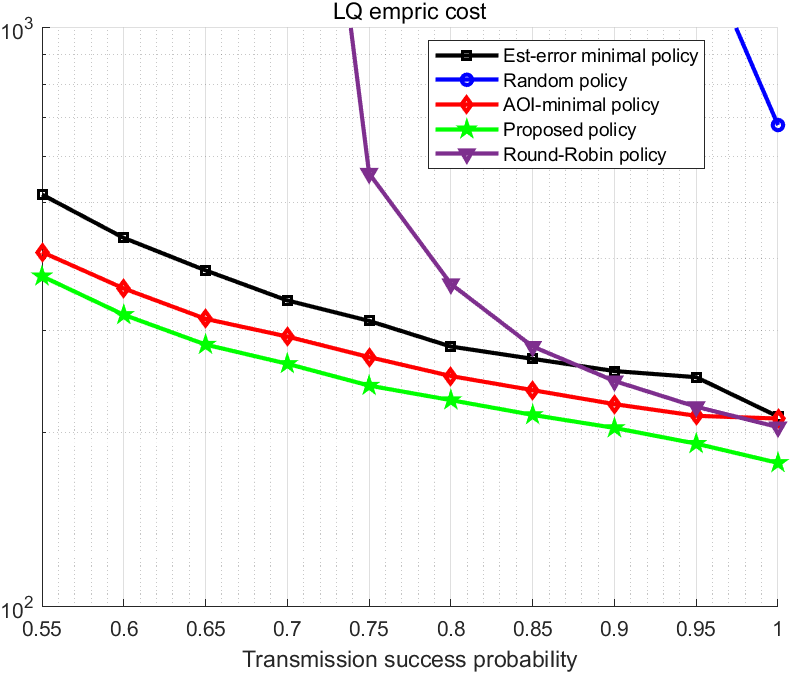}}
\caption{LQ cost comparison of the five types of scheduling policies with different channel states $p$ and $M=3$ wireless resources in time sensitive systems.}
\label{fig4}
\end{figure}
\begin{figure}[!t]
\centerline{\includegraphics[width=2.5in]{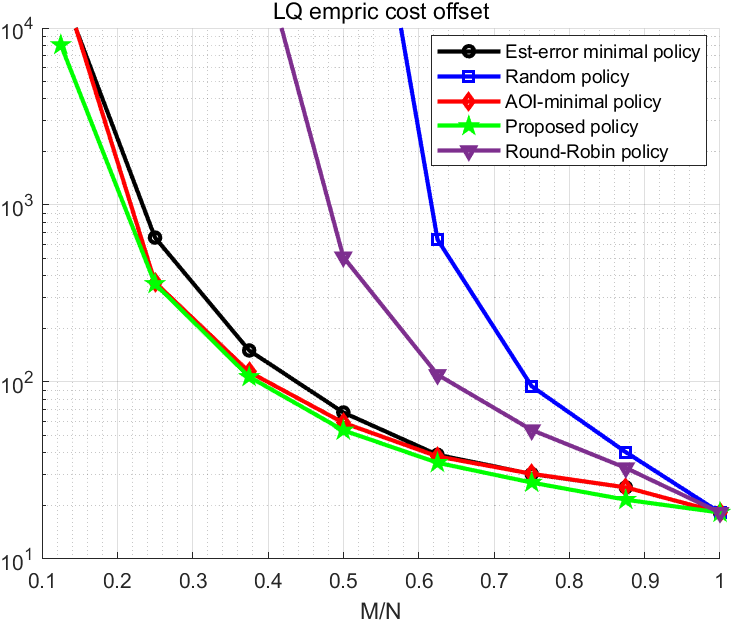}}
\caption{Cost offset comparison of the five types of scheduling policies with different wireless resources $M$ per time slot and transmission success rate $p=0.7$.}
\label{fig5}
\end{figure}
\begin{figure}[!t]
\centerline{\includegraphics[width=2.5in]{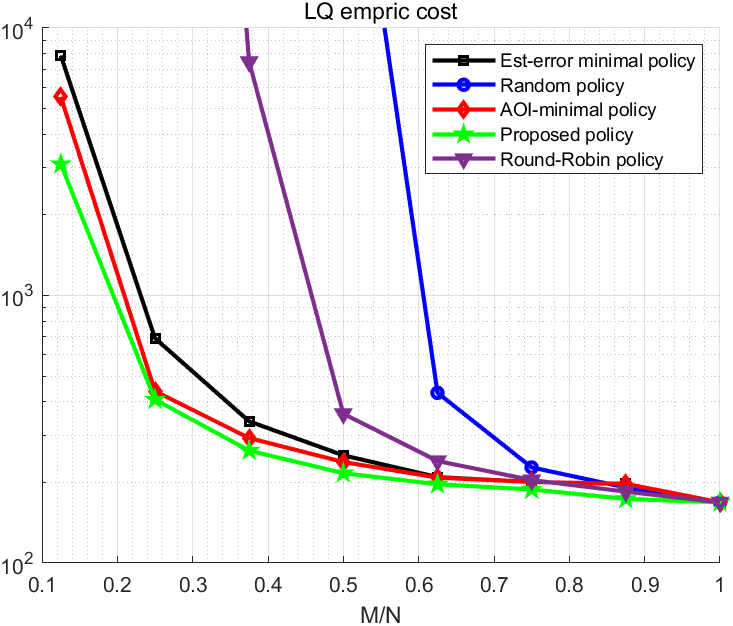}}
\caption{LQ cost comparison of the five types of scheduling policies with different wireless resources $M$ and transmission success rate $p=0.7$.}
\label{fig6}
\end{figure}
Figure \ref{fig3} and Figure \ref{fig5} compare linear quadratic empiric cost offset for AoI-minimal policy, Estimation error minimal policy, Round-Robin Policy, Random policy and our proposed policy which minimize real-time control offset of the whole systems in time sensitive scenarios. Different transmission success probabilities $p$ and resource ratios $\frac{M}{N}$ are considered,  and it can be observed that random policy has the worst control performance. With the increase of trasmission success probability and resource ratios, the cost offsets gradually reduce. When transmission resources are abundant, all policies perform well beacause almost every system can be served in every time slot. We obtain some gain compared to the existed policies in such systems no matter what transmission conditions are. %较低传输成功率，也就是通信条件较为恶劣的情况下，我们所提出的方法相对于AOI VOI方法有更高的性能优势，而通信条件较优的时候，AOI-VOI-我们提出的都有较好的性能。
\par
Figure \ref{fig4} and Figure \ref{fig6} show comparison of linear quadratic empiric cost for the five types of policys. Same as the cost offset, our policy has an advantage in performance. In the meantime, it is worth mentioning that there is a similar characterization of performance, which means we could optimize our proposed cost offset of such systems to indirectly optimize the LQ cost to some extent. It will be helpful for system co-design because sometimes it is hard to solve complex control cost optimization. %和gap相比，同样的，我们的方案具有更高的性能优势，于此同时，值得一提的是，Gap和Cost有相似的性能刻画，也就是说，某种程度上我们可以通过优化Gap来间接的优化Cost。因为Cost经常形式较为复杂，难以直接优化，通过这种方式能给系统设计提供一定的帮助。
\par
%Figure \ref{fig4} shows that  linear quadratic empiric cost gap with different resource ratio $\frac{M}{N}$.
%随着资源比例的提升，控制系统Cost逐渐降低，资源分配越高，不同方法的性能差异就约不明显，我们提出的方案在较低资源比例的情况下拥有较好的性能。
\par
\subsection{Large characteristic time systems}
On the other hand, we take stable systems with large characteristic time into consideration as well, which is shown in Table II. These stable systems are designed to have large characteristic time and therefore have better performance in the accuracy of control and noise immunity. %这种方案，让噪声占信号的比重降低，因此在一些场景下有更好的噪声平衡性能
We suppose that added feedback loops in some systems reduce the convergence speed and improves the noise immunity. For example, real part of eigenvalue of $\mathbf{A}_1$ is smaller than that of $\mathbf{A}_1+\mathbf{B}_1\mathbf{K}_1$.
\begin{table}[!t]
\caption{System Parameters and control laws of stable systems with large characteristic time}
\begin{center}
\begin{tabular}{|c|c|c|c|c|c|c|}
\hline
$i$&$\mathbf{A}_i$&$\mathbf{B}_i$&$\mathbf{Q}_i$&$\mathbf{P}_i$&$\mathbf{R}_i$&$\mathbf{K}_i$\\
\hline
1& $0.1*\mathbf{S}$&$\mathbf{I}$&$\mathbf{I}$&$\mathbf{I}$&$0.25*\mathbf{I}$&$0.8*\mathbf{I}$  \\
\hline
2& $0.1*\mathbf{S}$&$\mathbf{I}$&$\mathbf{I}$&$\mathbf{I}$&$0.25*\mathbf{I}$&$0.75*\mathbf{I}$  \\
\hline
3& $0.2*\mathbf{S}$&$\mathbf{I}$&$\mathbf{I}$&$\mathbf{I}$&$0.25*\mathbf{I}$&$0.5*\mathbf{I}$  \\
\hline
4& $0.2*\mathbf{S}$&$\mathbf{I}$&$\mathbf{I}$&$\mathbf{I}$&$0.25*\mathbf{I}$&$0.455*\mathbf{I}$  \\
\hline
5& $0.3*\mathbf{S}$&$\mathbf{I}$&$\mathbf{I}$&$\mathbf{I}$&$0.25*\mathbf{I}$&$-0.05*\mathbf{I}$  \\
\hline
6& $0.3*\mathbf{S}$&$\mathbf{I}$&$\mathbf{I}$&$\mathbf{I}$&$0.25*\mathbf{I}$&$-0.1*\mathbf{I}$  \\
\hline
7& $0.4*\mathbf{S}$&$\mathbf{I}$&$\mathbf{I}$&$\mathbf{I}$&$0.25*\mathbf{I}$&$-0.35*\mathbf{I}$  \\
\hline
8& $0.4*\mathbf{S}$&$\mathbf{I}$&$\mathbf{I}$&$\mathbf{I}$&$0.25*\mathbf{I}$&$-0.38*\mathbf{I}$  \\
\hline
\multicolumn{6}{l}{where $\mathbf{I}=\begin{aligned}
\left[
\begin{array}{cc}
1&0\\
0&1
\end{array}
\right]
\end{aligned}$  and $\mathbf{S}=\begin{aligned}
\left[
\begin{array}{cc}
1&0.2\\
-0.2&1
\end{array}
\right]
\end{aligned}$}
\end{tabular}
\label{tab1}
\end{center}
\end{table}
\begin{figure}[!t]
\centerline{\includegraphics[width=2.5in]{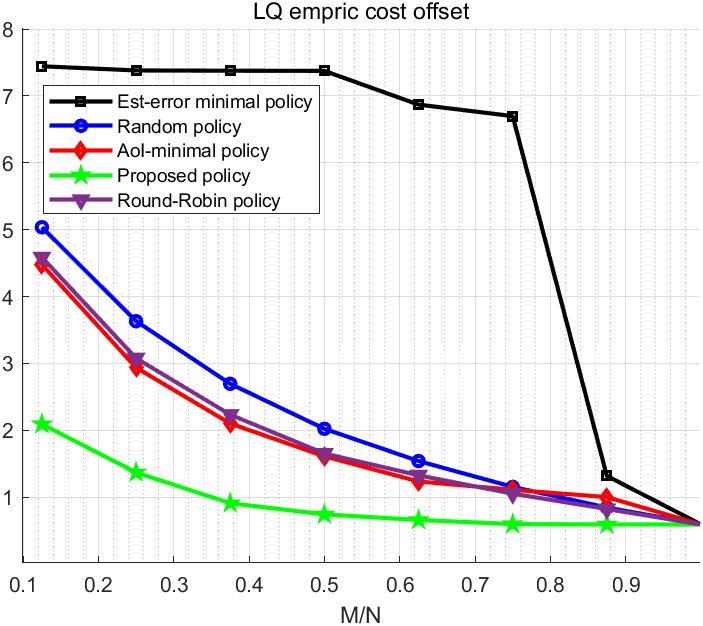}}
\caption{LQ cost comparison of the five types of scheduling policies with different resource allocation rates and 0.7 transmission success probabilities.}
\label{fig7}
\end{figure}
\begin{figure}[!t]
\centerline{\includegraphics[width=3.1in]{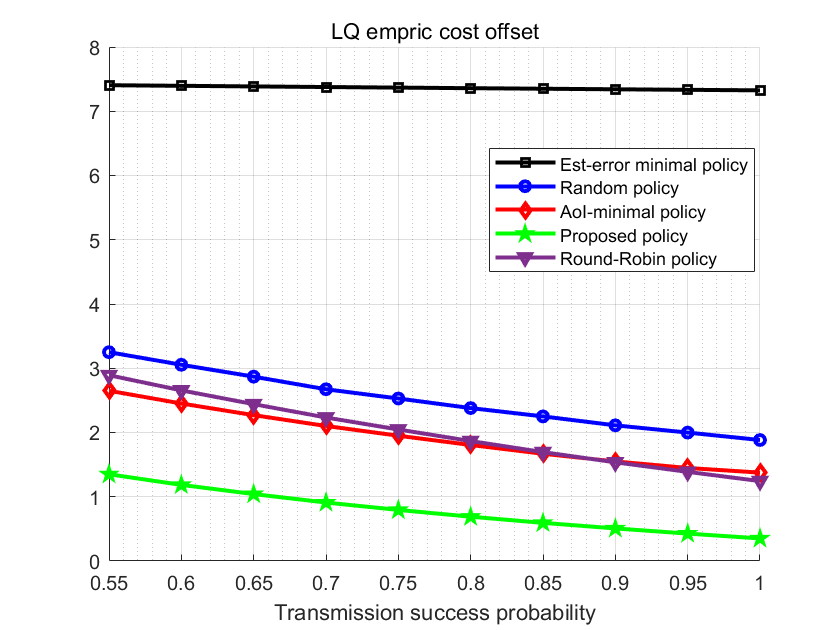}}
\caption{LQ cost comparison of the five types of scheduling policies with different transmission success rates and $M=3$ wireless resources.}
\label{fig8}
\end{figure}
\par
From Figure \ref{fig7} and Figure \ref{fig8}, we can see that compared to the time sensitive systems situation, our proposed policy obtains much more gains in cost offset criterion.  As cost offset criterion indicates the offset between actual situation and the ideal design, our proposed policy has better performance in mitigating non-ideal communication loss. Furthermore, Estimation error minimal policy has a deteriorating performance due to limitations of simply optimizing estimation error.
\section{Conclusion and Furture Work}
This paper has studied multi-loop wireless networked control systems which share a common wireless channel, presented a new metric criterion to descirbe the performance offset between situations with ideal communication and non-ideal communication, and has proposed a scheduling policy which at every time slot chooses subsystems with maximum expected offset cost to transmit. Simulation results show that compared to existed policies, our proposed policy may help to improve control performance in many different WNCSs cases.
Other future research directions include partial state transmission scheduling, distributed control transmission scheme design and control law design based on the process of AoI.
\bibliographystyle{plain}
\bibliography{cite}

\begin{thebibliography}{10}

\bibitem{ayan2019age}
Onur Ayan, Mikhail Vilgelm, Markus Kl{\"u}gel, Sandra Hirche, and Wolfgang
  Kellerer.
\newblock Age-of-information vs. value-of-information scheduling for cellular
  networked control systems.
\newblock In {\em Proceedings of the 10th ACM/IEEE International Conference on
  Cyber-Physical Systems}, pages 109--117. ACM, 2019.

\bibitem{cervin2008scheduling}
Anton Cervin and Toivo Henningsson.
\newblock Scheduling of event-triggered controllers on a shared network.
\newblock In {\em 2008 47th IEEE Conference on Decision and Control}, pages
  3601--3606. IEEE, 2008.

\bibitem{champati2019performance}
Jaya~Prakash Champati, Mohammad~H Mamduhi, Karl~H Johansson, and James Gross.
\newblock Performance characterization using aoi in a single-loop networked
  control system.
\newblock {\em arXiv preprint arXiv:1901.06694}, 2019.

\bibitem{garone2011lqg}
Emanuele Garone, Bruno Sinopoli, Andrea Goldsmith, and Alessandro Casavola.
\newblock Lqg control for mimo systems over multiple erasure channels with
  perfect acknowledgment.
\newblock {\em IEEE Transactions on Automatic Control}, 57(2):450--456, 2011.

\bibitem{hassibi1999control}
Arash Hassibi, Stephen~P Boyd, and Jonathan~P How.
\newblock Control of asynchronous dynamical systems with rate constraints on
  events.
\newblock In {\em Proceedings of the 38th IEEE Conference on Decision and
  Control (Cat. No. 99CH36304)}, volume~2, pages 1345--1351. IEEE, 1999.

\bibitem{henriksson2015multiple}
Erik Henriksson, Daniel~E Quevedo, Edwin~GW Peters, Henrik Sandberg, and
  Karl~Henrik Johansson.
\newblock Multiple-loop self-triggered model predictive control for network
  scheduling and control.
\newblock {\em IEEE Transactions on Control Systems Technology},
  23(6):2167--2181, 2015.

\bibitem{huang2019retransmit}
Kang Huang, Wanchun Liu, Yonghui Li, and Branka Vucetic.
\newblock To retransmit or not: Real-time remote estimation in wireless
  networked control.
\newblock {\em arXiv preprint arXiv:1902.07820}, 2019.

\bibitem{kadota2018scheduling}
Igor Kadota, Abhishek Sinha, Elif Uysal-Biyikoglu, Rahul Singh, and Eytan
  Modiano.
\newblock Scheduling policies for minimizing age of information in broadcast
  wireless networks.
\newblock {\em IEEE/ACM Transactions on Networking (TON)}, 26(6):2637--2650,
  2018.

\bibitem{Kaul2012real}
Gruteser~M Kaul~S, Yates~R.
\newblock Real-time status: How often should one update?
\newblock In {\em IEEE INFOCOM. IEEE, 2012}, pages 2731--2735. IEEE, 2012.

\bibitem{kim2003maximum}
Dong-Sung Kim, Young~Sam Lee, Wook~Hyun Kwon, and Hong~Seong Park.
\newblock Maximum allowable delay bounds of networked control systems.
\newblock {\em Control Engineering Practice}, 11(11):1301--1313, 2003.

\bibitem{li2013adaptive}
Husheng Li, Ju~Bin Song, and Qi~Zeng.
\newblock Adaptive modulation in networked control systems with application in
  smart grids.
\newblock {\em IEEE Communications Letters}, 17(7):1305--1308, 2013.

\bibitem{molin2009lqg}
Adam Molin and Sandra Hirche.
\newblock On lqg joint optimal scheduling and control under communication
  constraints.
\newblock In {\em Proceedings of the 48h IEEE Conference on Decision and
  Control (CDC) held jointly with 2009 28th Chinese Control Conference}, pages
  5832--5838. IEEE, 2009.

\bibitem{molin2014price}
Adam Molin and Sandra Hirche.
\newblock Price-based adaptive scheduling in multi-loop control systems with
  resource constraints.
\newblock {\em IEEE Transactions on Automatic Control}, 59(12):3282--3295,
  2014.

\bibitem{robinson2007sending}
CL~Robinson and PR~Kumar.
\newblock Sending the most recent observation is not optimal in networked
  control: Linear temporal coding and towards the design of a control specific
  transport protocol.
\newblock In {\em 2007 46th IEEE Conference on Decision and Control}, pages
  334--339. IEEE, 2007.

\bibitem{tarable2013anytime}
Alberto Tarable, Alessandro Nordio, Fabrizio Dabbene, and Roberto Tempo.
\newblock Anytime reliable ldpc convolutional codes for networked control over
  wireless channel.
\newblock In {\em 2013 IEEE International Symposium on Information Theory},
  pages 2064--2068. IEEE, 2013.

\bibitem{walsh2002stability}
Gregory~C Walsh, Hong Ye, and Linda~G Bushnell.
\newblock Stability analysis of networked control systems.
\newblock {\em IEEE transactions on control systems technology},
  10(3):438--446, 2002.

\bibitem{xu2017lqg}
Yang Xu, Karl-Erik {\AA}rz{\'e}n, Enrico Bini, and Anton Cervin.
\newblock Lqg-based control and scheduling co-design.
\newblock {\em IFAC-PapersOnLine}, 50(1):5895--5900, 2017.

\end{thebibliography}
\end{document}